\documentclass{llncs}
\usepackage{fullpage}

\usepackage{cite}
\usepackage{amsmath,amssymb}
\usepackage[latin1]{inputenc}
\usepackage{semantic}

\newcommand{\RR}{\mathbb{R}}

\newcommand{\NN}{\mathbb{N}}

\usepackage{comment,color}

\newtheorem{innermech}{Mechanism}
\newenvironment{mechanism}[1]
  {\renewcommand\theinnermech{#1}\innermech}
  {\endinnermech}

\usepackage{algpseudocode}

\mathlig{->}{\rightarrow} 
\mathlig{-->}{\longrightarrow} 
\mathlig{=>}{\Rightarrow} 
\mathlig{==>}{\Longrightarrow} %
\mathlig{<-}{\leftarrow} \mathlig{<--}{\longleftarrow} 
\mathlig{<=}{\Leftarrow} \mathlig{<==}{\Longleftarrow} %
\mathlig{<->}{\leftrightarrow} \mathlig{<-->}{\longleftrightarrow}
\mathlig{<=>}{\Leftrightarrow} \mathlig{<==>}{\Longleftrightarrow}%
\mathlig{...}{\ldots} 
\mathlig{....}{\ldots} 
\mathlig{==}{\equiv} 

\DeclareMathOperator*{\EE}{\mathbb{E}}

\begin{document}


\title{Social welfare in one-sided matchings: Random priority and beyond\thanks{The authors acknowledge support from the Danish National Research Foundation and The National Science
Foundation of China (under the grant 61061130540) for the Sino-Danish Center for the Theory of Interactive Computation,
within which this work was performed. The authors also acknowledge support from the Center for Research
in Foundations of Electronic Markets (CFEM), supported by the Danish Strategic Research Council. Jie Zhang also acknowledges support from ERC Advanced Grant 321171 (ALGAME).}
}

\author{
Aris Filos-Ratsikas\inst{1} \and S\o ren Kristoffer Stiil Frederiksen\inst{1} \and Jie Zhang\inst{2}
}

\institute{Department of Computer Science, Aarhus University \\
\email{\{filosra,ssf\}@cs.au.dk}\\
\and Department of Computer Science, University of Oxford\\
\email{jie.zhang@cs.ox.ac.uk}}

\maketitle
\begin{abstract}
We study the problem of approximate social welfare maximization (without money) in one-sided matching problems when agents have unrestricted cardinal preferences over a finite set of items. \emph{Random priority} is a very well-known \emph{truthful-in-expectation} mechanism for the problem. We prove that the approximation ratio of random priority is $\Theta(n^{-1/2})$ while no truthful-in-expectation mechanism can achieve an approximation ratio better than $O(n^{-1/2})$, where $n$ is the number of agents and items. Furthermore, we prove that the approximation ratio of all ordinal (not necessarily truthful-in-expectation) mechanisms is upper bounded by $O(n^{-1/2})$, indicating that random priority is asymptotically the best truthful-in-expectation mechanism and the best ordinal mechanism for the problem. 
\end{abstract}

\section{Introduction} \label{sec:intro}

We study the problem of approximate social welfare maximization (without money) in one-sided matching problems when agents have unrestricted cardinal preferences over a finite set of items. Specifically, each agent has a valuation function mapping items to real numbers, which can be arbitrary. These valuation functions should be interpreted as von Neumann-Morgenstern utility functions, i.e. they induce orderings on outcomes, and are standardly defined up to positive affine transformations (multiplication by a positive scalar and shift by a scalar). 

A (direct revelation) mechanism (without money) is a function $J$ mapping vectors of valuation functions (valuation profiles) to matchings, that is, allocations of items to agents such that each agent is allocated exactly one item. Mechanisms can also be randomized, and then the function $J$ is a random map. We will be interested in \emph{truthful} mechanisms, that is, mechanisms that do not give incentives to agents to misreport their valuation functions. For randomized mechanisms, we are interested in mechanisms that are \emph{truthful-in-expectation}, which means that no agent can increase its expected utility by misreporting. 

A very intuitive and well-studied truthful-in-expectation mechanism is \emph{random priority} (often also referred to as \emph{random serial dictatorship}), which first fixes a uniformly random ordering of the agents and then serially matches each agent to its most preferred item from the set of still unmatched items, based on that ordering. Random priority is an \emph{ordinal} mechanism, i.e., a mechanism that only depends on the ordering of items induced by the valuation functions and not the actual numerical values.

The goal we will look to maximize is the \emph{social welfare}, that is, the sum of agents' valuations for the items they are matched with in the outcome of the mechanism (and the expected social welfare for randomized mechanisms). We measure the performance of a mechanism by its approximation ratio, which is the worst ratio between the (expected) social welfare of the mechanism and the welfare of the optimal allocation, over all valuation profiles. It is easy to see that the mechanism that always outputs the optimal outcome is not truthful.

For a meaningful discussion on social welfare maximization for von Neumann-Morgenstern utilities, one has to fix a canonical representation of the valuation functions \cite{DG:10,GC:10, FM:13, BCHLPS:12}. These functions are usually represented in one of two canonical forms, \emph{unit-sum} \cite{BCHLPS:12,GC:10} (the valuations sum up to one)
or \emph{unit-range} \cite{ZHOU:90,BA:10,FM:13} (all valuations are in the interval $[0,1]$ with both $0$ and $1$ in the image of the function). Our main result is given by the following theorem, which holds for both normalizations.

\begin{theorem}\label{thm:main}
The approximation ratio of random priority is $\Theta(1/\sqrt{n})$. Furthermore, random priority is asymptotically the best truthful-in-expectation mechanism and the best ordinal (not necessarily truthful-in-expectation) mechanism for the problem.
\end{theorem} 

\noindent The theorem also holds for an extension to the unit-range representation, when $0$ is not required to be in the image of the function; we discuss how in Section \ref{sec:future}.

\subsection{Discussion and related work}

In the presence of incentives, the one-sided matching problem (often referred to as \emph{the assignment problem} or \emph{house allocation problem}) was originally defined in \cite{HZ:79} and has been studied extensively ever since \cite{ZHOU:90, BM:01, GC:10, SVE:99, DG:10}.  There are several surveys discussing the problem (as well as more general matching problems) \cite{AS:13,SU:11} and we refer the interested reader to those for a more detailed exposition.
Random priority is a folklore mechanism that solves the problem fairly (in the sense of anonymity) and satisfies some additional nice properties; it is truthful-in-expectation and \emph{ex-post Pareto efficient}. On the other hand, it is not \emph{ex-ante Pareto efficient}, i.e., there exists some matching (different than the outcome of random priority) for which all agents are at least as satisfied and one agent is strictly better off \emph{in expectation}. Most of the previous work in literature \cite{ZHOU:90,BM:01,HZ:79} has mainly been directed towards designing mechanisms with desired properties that achieve efficiency criteria for different Pareto efficiency notions and less towards whether truthful mechanisms achieve ``good levels'' of social welfare. This is a really important question to ask, especially since the notion of approximation ratio gives us a systematic way of comparing mechanisms or proving their limitations. The term \emph{approximate mechanism design without money} was used in \cite{PT:09} to describe problems where the goal is to approximately optimize some objective function, given the constraint of truthfulness. This approach has been adopted by a large body of computer science literature \cite{GC:10,FM:13,BCK:11,DG:10,PP:10,AN:10} and the approximation ratio is now considered to be the predominant measure of efficiency for truthful mechanisms.

Social welfare maximization is arguably a less natural objective when agents are endowed with von Neumann-Morgenstern utilities, because of adding up valuations after normalization. On the other hand, it is quite widespread in \emph{quasi-linear} settings. We strongly believe that considering the social welfare objective for von Neumann-Morgenstern utilities is just as natural and in fact there is a growing amount of literature that embraces the same idea and provides arguments to support it \cite{GC:10, FM:13, BCHLPS:12}. Finally, it is not difficult to see that no nontrivial approximation guarantees can be achieved by truthful-in-expectation mechanisms without any normalization \cite{GC:10,DG:10,BCHLPS:12}. 

A different approach, often encountered in literature, is to consider ordinal measures of efficiency. For example, Bhalgat et al \cite{BCK:11} calculate the approximation ratio of random priority when the objective is the maximization of \emph{ordinal} social welfare, a notion of welfare that they define based solely on ordinal information. Ordinal measures of efficiency have also been studied in terms of incentives and approximation ratios \cite{PP:10}. However, these measures do not encapsulate the ``socially desired'' outcome in the way that social welfare does, i.e., they do not necessarily maximize the aggregate happiness of individuals \cite{AN:10}. This is even more evident if one considers that the the standard assumption in social choice and economics theory is that such an underlying cardinal structure exists, even if agents are not asked to report it. In our setting, not reporting the full cardinal information corresponds to using ordinal mechanisms, which by our main theorem, is enough for achieving the (asymptotically) best approximation guarantees. On the other hand, cardinal reports are also often encountered in literature, with the pseudo-market mechanism of \cite{HZ:79} being a prominent example. Several cardinal mechanisms were also presented in \cite{GC:10,FM:13} for social welfare maximization and \cite{FT:10} for information elicitation. Our main theorem does not only prove the capabilities of random-priority but also the limitations of all truthful (including cardinal) mechanisms. In Section \ref{sec:future}, we show that for the special case of three agents and three items, a truthful-in-expectation mechanism strictly outpeforms all ordinal.

\section{Preliminaries} \label{sec:prelim}

Let $N=\{1,\ldots,n\}$ be a finite set of agents and $M=\{1,\ldots,n\}$ be a finite set of indivisible items. An \emph{outcome} is a matching of agents to items, that is, an assignment of items to agents where each agent gets assigned exactly one item. We can view an outcome $\mu$ as a vector $(\mu_1,\mu_2\ldots, \mu_n)$ where $\mu_i$ is the unique item matched with agent $i$. Let $O$ be the set of all outcomes. Each agent $i$ has a private valuation function mapping outcomes to real numbers that can be arbitrary except for two conditions; agents are indifferent between outcomes that match them to the same item and they are not indifferent between outcomes that match them to different items. The first condition implies that agents only need to specify their valuations for items instead of outcomes and hence the valuation function of an agent $i$ can be instead defined as a map $u_i:M \rightarrow \RR$ from items to real numbers. The second condition requires that valuation functions are injective, i.e., they induce a total ordering on the items. This is mainly for convenience, to avoid having to specify tie-breaking rules. As we discuss in Section \ref{sec:future}, all of our results extend to most natural tie-breaking rules. Valuation functions are standardly considered to be well-defined up to positive affine transformations, that is, for item $j: j\rightarrow\alpha u_i(j)+\beta$ is considered to be a different representation of $u_i$. The two standard ways to fix the canonical representation of $u_i$ in literature are \emph{unit-range}, i.e., $\max_j u_i(j)=1$ and $\min_j u_i(j)=0$ and \emph{unit-sum}, that is $\sum_j u_i(j)=1$.

Let $V$ be the set of all canonically represented valuation functions of an agent. Call $\mathbf{u}=(u_1,u_2,\ldots, u_n)$ a \emph{valuation profile} and let $V^n$ be the set of all valuation profiles with $n$ agents. A \emph{direct revelation mechanism} (without money) is a function $J: V^n\rightarrow  O$ mapping \emph{reported} valuation profiles to matchings. For a randomized mechanism, we define $J$ to be random map $J:V^n\rightarrow O$. Let $J(\mathbf{u})_i$ denote the restriction of the outcome of the mechanism to the $i$'th coordinate, which is the item assigned to agent $i$ by the mechanism.

We will be interested in \emph{truthful mechanisms}, that is, mechanisms that do not incentivize agents to report anything other than their true valuation functions. Formally, a mechanism $J$ is truthful if for each agent $i$ and all $\mathbf{u}=(u_i,u_{-i}) \in V^n$ and $\tilde{u}_i \in V$ it holds that
$u_i(J(u_i,u_{-i})_i) \geq u_i(J(\tilde{u}_i,u_{-i})_i)$,
where $u_{-i}$ denotes the valuation profile $\mathbf{u}$ without the $i$'th coordinate. In other words, if $u_i$ is agent $i$'s true valuation function, then it has no incentive to misreport. For randomized mechanisms, we say that a mechanism is \emph{truthful-in-expectation} if for each agent $i$ and all $\mathbf{u}=(u_i,u_{-i}) \in V^n$ and $\tilde{u}_i \in V$ it holds that
$\EE[u_i(J(u_i,u_{-i})_i)] \geq \EE[u_i(J(\tilde{u}_i,u_{-i})_i)]$.

A class of mechanisms that turns out to be important for our purposes is that of neutral and anonymous mechanisms. Formally, a mechanism is \emph{anonymous} if for any valuation profile $(u_1,u_2,\ldots,u_n)$, every agent $i$ and any permutation $\pi:N\rightarrow N$ it holds that $J(u_1,u_2,\ldots,u_n)_i = J(u_{\pi(1)},u_{\pi(2)},\ldots,u_{\pi(n)})_{\pi(i)}$. By this definition, in an anonymous mechanism, agents with exactly the same valuation functions must have the same probabilities of receiving each item. Similarly, a mechanism is \emph{neutral} if for any valuation profile $(u_1,u_2,\ldots,u_n)$, every item $j$ and any permutation $\sigma:M\rightarrow M$ it holds that $J(u_1,u_2,\ldots,u_n)_i = \sigma^{-1}(J(u_1  \circ \sigma, u_2 \circ \sigma,\ldots, u_n \circ \sigma)_i)$, i.e., the mechanism is invariant to the indices of the items.

We will consider both ordinal and cardinal mechanisms. A mechanism $J$ is \emph{ordinal} if for any $i$, any valuation profile $\mathbf{u}=(u_i,u_{-i})$ and any valuation function $u_i'$ such that for all $j,j' \in M$, $u_i(j)<u_i(j') \Leftrightarrow u_i'(j)<u_i'(j')$, it holds that $J(u_i,u_{-i})=J(u_i',u_{-i})$. A mechanism for which the above does not necessarily hold is \emph{cardinal}. Informally, ordinal mechanisms operate solely based on the \emph{ordering} of items induced by the valuation functions and not the actual numerical values themselves, while cardinal mechanisms take those numerical values into account when outputting an outcome. 

We measure the performance of a mechanism by its approximation ratio,
\begin{equation*}
ar(J) = \inf_{\mathbf{u} \in V^n}\frac{\sum_{i=1}^{n}u_i(J(\mathbf{u})_i)}{\max_{\mu \in O}\sum_{i=1}^{n}u_i(\mu_i)}
\end{equation*}

\noindent The quantity $\sum_{i=1}^{n}u_i(J(\mathbf{u})_i)$ is called the \emph{social welfare} of mechanism $J$ on the valuation profile $\mathbf{u}$ and $\max_{\mu \in O}\sum_{i=1}^{n}u_i(\mu_i)$ is the social welfare of the optimal matching. For ease of notation, let $w^{*}(\mathbf{u})=\max_{\mu \in O}\sum_{i=1}^{n}u_i(\mu_i)$. For the case of randomized mechanisms, we will be interested in the \emph{expected social welfare} $\EE\left[\sum_{i=1}^{n}u_i(J(\mathbf{u})_i)\right]$ of mechanism $J$ and the approximation ratio is defined accordingly. 

Next we will state a lemma that will be useful for our proofs. These kinds of lemmas are standard in literature (e.g. see \cite{GC:10, FM:13}). The lemma implies that when trying to prove upper bounds on the approximation ratio of mechanisms, it suffices to consider mechanisms that are anonymous. 

\begin{lemma}\label{lem:anonymous}
For any mechanism $J$, there exists an anonymous mechanism $J'$ such that $ar(J')\geq ar(J)$. Furthermore, if $J$ is truthful (for deterministic mechanisms) or truthful-in-expectation (for randomized mechanisms) then it holds that $J'$ is truthful-in-expectation.
\end{lemma} 

\begin{proof}
Let $J'$ be the mechanism that given any valuation profile $\mathbf{u}$ applies a uniformly random permutation to the set of agents and then applies $J$ on $\mathbf{u}$. The mechanism is clearly anonymous. Furthermore, since $\mathbf{u}$ is a valid input to $J$, the approximation ratio of $J'$ can not be worse than that of $J$, since the approximation ratio is calculated over all possible valuation profiles. For the same reason, if $J$ is truthful (or truthful-in-expectation) and since the permutation is independent of the reports, $J'$ is truthful-in-expectation. \hfill $\square$ 
\end{proof}

We will particularly be interested in the mechanism \emph{random priority}. Random priority fixes an ordering of the agents uniformly at random and then lets them pick their most preferred items from the set of available items based on this ordering. Note that random priority is truthful-in-expectation, ordinal, anonymous and neutral. We conclude the section with the following lemma. Similar lemmas have been proved in literature (e.g. see Lemma 1 in \cite{BM:01}, for a slightly more general statement).

\begin{lemma}\label{lem:RPopt}
For any valuation profile $\mathbf{u}$, the optimal allocation on $\mathbf{u}$ is a possible outcome of random priority.    
\end{lemma}

\begin{proof}
First, suppose that no agent is matched with its most preferred item in the optimal allocation. Then there must exist agents $i_1,...,i_k$ such that for each $l$, agent $i_{l+1}$ is matched with agent $i_{l}$'s most preferred item and agent $i_1$ is matched with agent $i_k$'s most preferred item. By swapping items along this cycle, all agents are better off and the allocation is not optimal. 

Now consider any valuation profile $\mathbf{u}$. Since there exists an agent $j$ that is matched with its most preferred item $j$ in the optimal allocation for $\mathbf{u}$, random priority could pick this agent first. If we reduce $\mathbf{u}$ by removing the agent $i$ and item $j$, we obtain a smaller valuation profile $\mathbf{u'}$ where the optimal allocation is the same as in $\mathbf{u}$ but without agent $i$ and item $j$. Then by inductively applying the same argument, the lemma follows. \hfill $\square$ 
\end{proof}

\section{Unit-range valuation functions} \label{sec:unitrange}

In this section, we assume that the representation of the valuation functions is unit-range. It will be useful to consider a  special class of valuation functions $C_\epsilon$ that we will refer to as \emph{quasi-combinatorial valuation functions}, a straightforward adaptation of a similar notion in \cite{FM:13}. Informally, a valuation function is quasi-combinatorial if the valuations of each agent for every item are close to $1$ or close to $0$ (the proximity depends on $\epsilon$). Formally, 
\begin{equation*}
C_\epsilon =\left\{ u \in V| u(M) \subset [0,\epsilon) \cup (1-\epsilon,1]\right\},
\end{equation*}

\noindent where $u(M)$ is the image of the valuation function $u$. Let $C_\epsilon^n \subseteq V^n$ be the set of all valuation profiles with $n$ agents whose valuation functions are in $C_\epsilon$. The following lemma implies that when we are trying to prove a lower bound on the approximation ratio of random priority, it suffices to restrict our attention to quasi-combinatorial valuation profiles $C_\epsilon^n \subseteq V^n$ for any value of $\epsilon$. 

\begin{lemma}\label{lem:zeroone}
Let $J$ be an ordinal, anonymous and neutral randomized mechanism for unit-range representation, and let $\epsilon>0$. Then 
\[  ar(J) = \inf_{{\bf u} \in C_\epsilon^n} \frac{\EE[\sum_{i=1}^{n}u_i(J(\mathbf{u})_i)]}{w^{*}(\mathbf{u})}. \]
\end{lemma}

\begin{proof}
Since $J$ is anonymous and neutral, we can assume that the optimal matching is $\mu^{*}$ where $\mu^{*}$ is the matching with $\mu^{*}_i=i$ for every agent $i \in N$. Given this, then for any valuation profile $\mathbf{u}$, define 
\[ g(\mathbf{u}) = \frac{\EE[\sum_{i=1}^{n}u_i(J(\mathbf{u})_i)]}{\sum_{i=1}^{n}u_i(\mu^{*}_i)}.\] 
\noindent Because of this, the approximation ratio can be written as $ar(J) = \inf_{\mathbf{u} \in V^n} g(\mathbf{u})$. Now since $C_\epsilon^n \subseteq V^n$, the lemma follows from the following claim:
\begin{equation*}
\text{For all } \mathbf{u} \in V^n \text{ there exists } \mathbf{u'} \in C_\epsilon^n \text{ such that } g(\mathbf{u'}) \leq g(\mathbf{u})
\end{equation*}
We will prove the claim by induction in $\sum_{i=1}^{n} \#\{ u_i(M) \cap [\epsilon,1-\epsilon] \}$.

\ \newline 
\emph{Induction basis:} Since $\sum_{i=1}^{n} \#\{ u_i(M) \cap [\epsilon,1-\epsilon] \} = 0$, one can clearly see that $u_i \in C_\epsilon$ for all $i \in N$. So, for this case, let $\mathbf{u'}=\mathbf{u}$.

\ \newline 
\emph{Induction step:} Consider a profile $\mathbf{u} \in V^n$ with $\sum_{i=1}^{n} \#\{ u_i(M) \cap [\epsilon,1-\epsilon]\} > 0$. Clearly, there exists an $i$ such that $\#\{ u_i(M) \cap [\epsilon,1-\epsilon]\} > 0$. By this fact, there exist $l,r \in [\epsilon,1-\epsilon]$, such that $l \leq r$, $u_i(M) \subset [0,\epsilon) \cup [l,r] \cup (1-\epsilon,1]$ and $\{l,r\} \subseteq u_i(M)$.

Let $\bar{l}$ be the largest number such that $\bar{l} \in [0,\epsilon)$ and $\bar{l} \in u_i(M)$. Similarly, let $\bar{r}$ be the smallest number such that $\bar{r} \in (1-\epsilon,1]$ and $\bar{r} \in u_i(M)$. Note that both those numbers exist, since $\{0,1\} \subseteq u_i(M)$. Now let $\tilde{l} = \frac{\bar{l} + \epsilon}{2}$, and $\tilde{r} =\frac{\bar{r} + 1 - \epsilon}{2}$

Now, for any $x \in [\tilde{l}-l,\tilde{r}-r]$, define a valuation function $u_i^x \in V$ as follows:
\begin{equation*}
u_i^x(j) = \begin{cases} 
u_i(j), &\mbox{for } j \notin u^{-1}_i\left([\epsilon,1 - \epsilon]\}\right) \\ 
u_i(j)+x, & \mbox{for } j \in u^{-1}_i\left([\epsilon,1 - \epsilon]\}\right). \end{cases} 
\end{equation*} 

This is still a valid valuation function, since by the choice of the interval $[\tilde{l}-l,\tilde{r}-r]$, there are no ties in the image of the function. Let $(u_i^x,\mathbf{u}_{-i})$ be the valuation profile where all agents have the same valuation functions as in $\mathbf{u}$ except for agent $i$, who has valuation function $u_i^x$. Define the following function $f:x \rightarrow g\left(\left(u_i^x,\mathbf{u}_{-i}\right)\right)$. Since $J$ is ordinal, by the definition of function $g$, we can see that $f$ on the domain $[\tilde{l}-l,\tilde{r}-r]$ is a fractional linear function $x\rightarrow (ax+b)/(cx+d)$ for some $a,b,c,d, \in \RR$. Since $f$ is defined on the whole interval $[\tilde{l}-l,\tilde{r}-r]$, it is either monotonically increasing, monotonically decreasing or constant in the interval. If $f$ is monotonically increasing, let $\tilde{\mathbf{u}} = (u^{\tilde{l}-l},\mathbf{u}_{-i})$, otherwise let $\tilde{\mathbf{u}} = (u^{\tilde{r}-r},\mathbf{u}_{-i})$. Clearly, $g(\tilde{\mathbf{u}}) \leq g(\mathbf{u})$ and 
\[\sum_{i=1}^{n} \#\{ \tilde{u}_i(M) \cap [\epsilon,1-\epsilon] \} < \sum_{i=1}^{n} \#\{ u_i(M) \cap [\epsilon,1-\epsilon] \}.\] Then, apply the induction hypothesis on $\mathbf{\tilde{u}}$. This completes the proof. \hfill $\square$
\end{proof}

The lemma formalizes the intuition that because the mechanism is ordinal, the worst-case approximation ratio is encountered on extreme valuation profiles. 

For the unit-range representation, Theorem \ref{thm:main} is given by the following lemmas.

\begin{lemma} For unit-range representation, $ar(RP)= \Omega\left(n^{-1/2}\right)$.\label{lem:RPlower}
\end{lemma}

\begin{proof}
Because of Lemma \ref{lem:zeroone}, for the purpose of computing a lower bound on the approximation ratio of random priority, it is sufficient to only consider quasi-combinatorial valuation profiles. Let $\epsilon \leq 1/n^3$. Then, there exists $k \in \NN$ such that 
\[|k-w^{*}(\mathbf{u})|\leq \frac{1}{n^2}, \] 
\noindent where $w^{*}(\mathbf{u})$ is the social welfare of the maximum weight matching on valuation profile $\mathbf{u}$. Since random priority can trivially achieve an expected welfare of $1$ (for any permutation the first agent will be matched to its most preferred item), we can assume that $k\geq\sqrt{n}$, otherwise we are done. 
\noindent Note that the maximum weight matching $\mu^{*} \in O$ assigns $k$ items to agents with $u_i(\mu_i) \in (1-\epsilon,1]$. Since random priority is anonymous and neutral, without loss of generality we can assume that these agents are $\{1,\ldots,k\}$ and for every agent $j \in N$, it holds that $\mu^{*}_j=j$. Thus $u_j(j) \in (1-\epsilon,1]$ for $j=1,\ldots,k$ and $u_j(j) \in [0,\epsilon)$ for $j=k+1,\ldots,n$.  

Consider any run of random priority; one agent is selected in each round. Let $l \in \{0,\ldots,n-1\}$ be any of the $n$ rounds. We will now define the following sets:
\begin{align*} 
U_l &= \{j \in \{1,\ldots,n\}| \textrm{ agent } j \textrm{ has not been selected prior to round } l \} \\
G_l &= \{j \in U_l | u_j(j) \in (1-\epsilon,1] \textrm{ and item } j \textrm{ is still unmatched} \} \\
B_l &= \{j \in U_l | u_j(j) \in [0,\epsilon) \textrm{ or item } j \textrm{ has already been matched to some agent} \}
\end{align*} 
These three families of sets should be interpreted as three sets that change over the course of the execution of random priority. $U_l$ is the set of agents yet to be matched, which is partitioned into $G_l$, the set of ``good'' agents, that guarantee a welfare of almost $1$ when picked, and $B_l$, the set of ``bad'' agents, that do not guarantee any contribution to the social welfare. For the purpose of calculating a lower bound, we will simply bound the sizes of the sets in these families. Obviously, $G_0=\{1,\ldots,k\}$ and $B_0=\{k+1,\ldots,n\}$. 

The probability that an agent $i \in G_l$ is picked in round $l$ of random priority is $|G_l|/(|G_l|+|B_l|)$, whereas the probability that an agent $i \in B_l$ is picked is $|B_l|/(|G_l|+|B_l|)$. By the discussion above, we can assume that whenever an agent from $G_l$ is picked its contribution to the social welfare is at least $1-\epsilon$ whereas the contribution from an agent picked from $B_l$ is less than $\epsilon$. In other words, the expected contribution to the social welfare from round $l$ is at least $|G_l|/(|G_l|+|B_l|)-\epsilon$.  

We will now upper bound $|G_l|$ and lower bound $|B_l|$ for each $l$. Consider round $l$ and sizes $|G_l|$ and $|B_l|$. First suppose that some agent $i$ from $G_l$ is picked and the agent is matched with item $j$. If $j \neq i$ and agent $j$ is in $G_l$, then $|G_{l+1}| =|G_l| -2$ and $|B_{l+1}| = |B_l| + 1$, since agent $j$ no longer has its item from the optimal allocation available and so agent $j$ is in $B_{l+1}$. On the other hand, if $j = i$ or agent $j$ is in $B_l$ then $|G_{l+1}| =|G_l| -1$ and $|B_{l+1}| = |B_l|$. In either case, $|G_{l+1}| \geq |G_l| - 2$ and $|B_{l+1}| \leq |B_l| +1$. Intuitively, the picked agent might take away some item from a good agent and turn it into a bad agent. 

Now suppose that agent $i$ from $B_l$ is picked and the agent is matched with item $j$. If agent $j$ is in $G_l$ then $|G_{l+1}| =|G_l| -1$ and $|B_{l+1}| = |B_l|$, since agent $j$ no longer has its item from the optimal allocation available and so agent $j$ is in $B_{l+1}$. On the other hand, if agent $j$ is in $B_l$ then $|G_{l+1}| =|G_l|$ and $|B_{l+1}| = |B_l| - 1$. In either case, $|G_{l+1}| \geq |G_l| -2$ and $|B_{l+1}| \leq |B_l|  +1 $.



To sum up, in each round $l$ of random priority, we can assume the size of $B_l$ increases by at most $1$ and the size of $G_l$ decreases by at most $2$.  Given this and that $|G_0| = k$ and $|B_0| = n-k$ and that $|G_l|>0$ for $l \leq \lfloor k/2 \rfloor$, we get

\begin{align*}
\EE\left[\sum_{i=1}^{n}u_i(RP(\mathbf{u})_i)\right] \geq \sum_{l=0}^n \left(\frac{|G_l|}{|G_l|+|B_l|} - \epsilon \right) \geq \sum_{l=0}^{\left\lfloor \frac{k}{2} \right\rfloor}  \frac{k-2l}{n-l} - n\epsilon
\end{align*}
and the ratio is
\begin{align*}
\frac{\EE\left[\sum_{i=1}^{n}u_i(RP(\mathbf{u})_i)\right]}{w^{*}(\mathbf{u})} 
&\geq \frac{\sum_{l=0}^{\left\lfloor \frac{k}{2} \right\rfloor} \frac{k-2l}{n-l}-n\epsilon}{k+\frac{1}{n^2}}
\geq \frac{\sum_{l=0}^{\left\lfloor \frac{k}{2} \right\rfloor} \frac{k-2l}{n-l}-n\epsilon}{2k} \\
 &= \sum_{l=0}^{\left\lfloor \frac{k}{2} \right\rfloor} \frac{1-\frac{2l}{k}}{2(n-l)}-\frac{n\epsilon}{2k}
 > \sum_{l=0}^{\left\lfloor \frac{k}{2} \right\rfloor} \frac{1-\frac{2l}{k}}{2n}-\frac{n\epsilon}{2k} \geq \frac{k-11}{8n}-\frac{n\epsilon}{2k}.
\end{align*}

The bound is clearly minimum when $k$ is minimum, that is, $k=\sqrt{n}$. Since this bound holds for any $\mathbf{u} \in C_\epsilon^n$, we get 

\begin{equation*}
ar(RP) = \inf_{\mathbf{u} \in C_\epsilon^n}\frac{\EE[\sum_{i=1}^{n}u_i(J(\mathbf{u})_i)]}{w^{*}(\mathbf{u})}
\geq \frac{\sqrt{n}-11}{8n} -\frac{n\epsilon}{2\sqrt{n}}. 
\end{equation*} 

We can choose $\epsilon$ so that the approximation ratio is at least $\frac{1}{20\sqrt{n}}$ for $n\geq 400$ and for $n \leq 400$, the bound holds trivially since random priority matches at least one agent with its most preferred item. \hfill $\square$  
\end{proof}

Next, we state the following lemma about ordinal mechanisms.

\begin{lemma} Let $J$ be any ordinal mechanism for unit-range representation. Then $ar(J) = O\left(n^{-1/2}\right)$.\label{lem:ordinalupper}
\end{lemma}

\begin{proof}
Let $\mathbf{u}=(u_1,u_2,\ldots,u_n)$ be the valuation profile where:
\begin{align*}
u_i(j) &= \begin{cases} 1-\frac{j-1}{n}  &\mbox{for } 1 \leq j \leq i \\
\frac{n-j}{n^2}  & \mbox{otherwise }\end{cases} \ \ &\forall i\in \{1,\ldots,\lfloor\sqrt{n}\rfloor \} \\
u_i(j) &= \begin{cases} 1 &\mbox{for } j=1 \\
\frac{n-j}{n^2} \ \ \ \ \ \  & \mbox{otherwise }\end{cases} \ \  &\forall i \in \{\lfloor\sqrt{n}\rfloor+1,\ldots,n\}  
\end{align*}

By Lemma \ref{lem:anonymous} we can assume that $J$ is anonymous. Notice that the valuation profile is \emph{ordered}, i.e., $u_i(j)>u_i(j')$ whenever $j<j'$ for all $j,j' \in M$ and all $i \in N$. Thus, any anonymous and ordinal mechanism on input $\mathbf{u}$ must output a uniformly random matching, that is, the probability that agent $i$ is matched with item $j$ is the same for all agents $i$, for every $j \in M$. The expected welfare of the mechanism on valuation profile $\mathbf u$ will be
\begin{eqnarray*}
\frac{1}{n}\sum_{i=1}^n\sum_{j=1}^n u_i(j) &\leq& \frac{1}{n}\left[\sum_{i=1}^{\lfloor \sqrt{n}\rfloor} \left(i+\frac{n-i}{n}\right) + \sum_{i=\lfloor \sqrt{n}\rfloor+1}^n \left(1+\frac{n-1}{n}\right)\right] \\
&\leq& 4+\frac{1}{2\sqrt{n}} \leq 5,
\end{eqnarray*}
where in the above expression, we upper bound each term $\frac{n-j}{n^2}$ by $\frac{1}{n}$ and each term $1-\frac{j}{n}$ by $1$. 

On the other hand, the social welfare of the maximum weight matching is
\begin{eqnarray*}
\sum_{i=1}^{\lfloor{\sqrt{n}}\rfloor} \left(1-\frac{i-1}{n}\right) + \sum_{i=\lfloor \sqrt{n} \rfloor+1}^n \frac{n-i}{n^2} \geq \sum_{i=1}^{\lfloor{\sqrt{n}}\rfloor} \left(1-\frac{i-1}{n}\right) \geq \lfloor{\sqrt{n}}\rfloor-1 \geq \frac{\sqrt{n}}{4}. 
\end{eqnarray*}

Where the final inequality holds for $n\geq 4$, the approximation ratio then is at most $\frac{20}{\sqrt{n}}$ for $n\geq 4$, and the bound holds trivially for $n<4$. \hfill $\square$
\end{proof}

Our final lemma provides a matching upper bound on the approximation ratio of any truthful-in-expectation mechanism.

\begin{lemma} Let $J$ be a truthful-in-expectation mechanism for unit-range representation. Then $ar(J)=O\left(n^{-1/2}\right)$. \label{lem:truthfulupper}
\end{lemma}
\begin{proof}
By Lemma \ref{lem:anonymous}, we can assume that Mechanism $J$ is anonymous. Let $k\ge 2$ be a parameter to be chosen later and let $\mathbf{u}=(u_1,u_2,\ldots, u_n)$ be the valuation profile where
\begin{align*}
u_i(j) &= \begin{cases} 1, &\mbox{for } j = i \\
\frac{2}{k}-\frac{j}{n}, &\mbox{for } 1\leq j \leq k+1, j \neq i \\
\frac{n-j}{n^2}, &\mbox{otherwise }\end{cases} \ \ &\forall i \in \{1,\ldots,k+1\} \\
u_i(j) &= \begin{cases} 1, &\mbox{for } j=1 \\
\frac{2}{k}-\frac{j}{n}, &\mbox{for } 2 \leq j \leq k+1 \\
\frac{n-j}{n^2}, &\mbox{otherwise }
\end{cases}  \ \  &\forall i \in \{k+2,\ldots,n\}  
\end{align*}

For $i=2,\ldots,k+1$, let $\mathbf{u^i} = (u_i',u_{-i})$ be the valuation profile where all agents besides agent $i$ have the same valuations as in $\mathbf{u}$ and $u_i'=u_{k+2}$. Note that when agent $i$ on valuation profile $\mathbf{u^i}$, reports $u_i$ instead of $u_i'$, the resulting valuation profile is $\mathbf{u}$. Since $J$ is anonymous and $u_i'=u_1=u_{k+2}=\ldots=u_n$, then agent $i$ receives at most a uniform lottery among these agents on valuation profile $\mathbf{u^i}$ and so it holds that 
\begin{eqnarray*}
\EE[u_i'(J(\mathbf{u^i})_i)] &\leq& \frac{1}{n-k+1}+\sum_{j=2}^{k+1}\frac{1}{n-k+1}\left(\frac{2}{k}-\frac{j}{n}\right) + \sum_{j=k+2}^{n} \frac{1}{n-k+1}\cdot \frac{n-j}{n^2} \\
&\leq& \frac{4}{n-k+1}
\end{eqnarray*}

Next observe that since $J$ is truthful-in-expectation, agent $i$ should not increase its expected utility by misreporting $u_i$ instead of $u_i'$ on valuation profile $\mathbf{u^i}$, that is,
\begin{equation}
\EE[u_i'(J(\mathbf{u^i})_i)] \geq \EE[u_i'(J(\mathbf{u})_i)] \label{truthfulness}
\end{equation}

For all $i=2,\ldots,k+1$, let $p_i$ be the probability that $J(\mathbf{u})_i=i$. Then, it holds that 
\begin{equation*}
\EE[u_i'(J(\mathbf{u})_i)] \geq p_i\left(\frac{2}{k}-\frac{i}{n}\right) \geq p_i\left(\frac{2}{k}-\frac{k+1}{n}\right) 
\end{equation*}

and by Inequality (\ref{truthfulness}) we get
\begin{eqnarray*}
&& p_i\left(\frac{2}{k}-\frac{k+1}{n}\right) \leq \frac{4}{n-k+1} \\
&=>& p_i \leq \frac{4}{n-k+1}\cdot \frac{kn}{2n-k(k+1)} \leq \frac{4}{n-k}\cdot \frac{kn}{2n-(k+1)^2} 
\end{eqnarray*}   

Let $p=\frac{4}{n-k} \cdot \frac{kn}{2n-(k+1)^2}$, i.e. for all $i$, $p_i \leq p$. We will next calculate an upper bound on the expected social welfare achieved by $J$ on valuation profile $\mathbf{u}$. 

For item $j=1$, the contribution to the social welfare is upper bounded by $1$. Similarly, for each item $j=k+2,\ldots,n$, its contribution to the social welfare is upper bounded by $1/n$. Overall, the total contribution by item 1 and items $k+2,\ldots,n$ will be upper bounded by $2$. 

We next consider the contribution to the social welfare from items $j=2,\ldots,k+1$. Define the random variables
\[
    X_j= 
\begin{cases}
    1,& \text{if } J(\mathbf{u})_j=j\\
    \frac{2}{k}-\frac{j}{n},              & \text{otherwise}
\end{cases}
\]
  
The contribution from items $j=2,\ldots,k+1$ is then $\sum_{j=2}^{k+1} X_j$ and so we get
\[ \EE\left[\sum_{j=2}^{k+1} X_j\right] = \sum_{j=2}^{k+1} \EE\left[X_j\right] \leq \sum_{j=2}^{k+1} \left(p + \frac{2}{k}-\frac{j}{n}\right) \leq kp+2 
\]

Overall, the expected social welfare of mechanism $J$ is at most $4+pk$ while the social welfare of the optimal matching is $k+1 + \sum_{i=k+2}^n \frac{n-i}{n^2}$ which is at least $k$. Since $p=\frac{4}{n-k} \cdot \frac{kn}{2n-(k+1)^2}$, the approximation ratio of $J$ then is
\begin{equation*}
ar(J)\le \frac{4+pk}{k}=\frac{4}{k}+\frac{4}{n-k} \cdot \frac{kn}{2n-(k+1)^2}
\end{equation*}

Let $k=\lfloor \sqrt{n} \rfloor-1$ and note that $\sqrt{n}-2 \le k \le \sqrt{n}-1$. Then,

\begin{eqnarray*}
ar(J)&\le& \frac{4}{k}+\frac{4}{n-k} \cdot \frac{kn}{2n-(k+1)^2} 
\le \frac{4}{\sqrt{n}-2}+\frac{4}{n-\sqrt{n}+1} \cdot \frac{(\sqrt{n}-1)n}{2n-(\sqrt{n})^2} \\
&\leq& \frac{4}{\sqrt{n}-2}+\frac{4}{\sqrt{n}}
\leq \frac{12}{\sqrt{n}}+\frac{4}{\sqrt{n}} 
= \frac{16}{\sqrt{n}},
\end{eqnarray*}
The last inequality holds for $n\geq 9$ and for $n < 9$ the bound holds trivially. This completes the proof.\hfill $\square$
\end{proof}

\section{Unit-sum valuation functions}\label{sec:unitsum}

In this section, we assume that the representation of the valuation functions is unit-sum. We prove Theorem \ref{thm:main} using the following three lemmas. The first lemma provides a lower bound on the approximation ratio of random priority.

\begin{lemma}\label{lem:RPlowerUS}
For unit-sum representation, $ar(RP) = \Omega\left(n^{-1/2}\right)$.
\end{lemma}  
 
\begin{proof}
Let $\mathbf{u}$ be any unit-sum valuation profile and let $C$ be the constant in the bound from Lemma \ref{lem:RPlower}. Suppose first that the $w^{*}(\mathbf{u})< 4 \sqrt{n} / C$. We will show that random priority guarantees an expected social welfare of $1$, which proves the lower bound for this case. Consider any agent $i$ and notice that in random priority, the probability that the agent is picked by the $l$'th round is $l/n$, for any $1\leq l \leq n$ and hence the probability of the agent getting one of its $l$ most preferred items is at least $l/n$. Let $u_i^l$ be agent $i$'s valuation for its $l$'th most preferred item; agent $i$'s expected utility for the first round is then at least $u_i^1/n$. For the second round, in the worst case, agent $i$'s most preferred item has already been matched to a different agent and so the expected utility of the agent for the first two rounds is at least $u_i^1/n + u_i^2/n$. By the same argument, agent $i$'s expected utility after $n$ rounds is at least $\sum_{i=1}^n u_i^l/n = 1/n$. Since this holds for each of the $n$ agents, the expected social welfare is at least $1$.

Suppose now $w^{*}(\mathbf{u})\geq 4\sqrt{n} / C$. We will transform $\mathbf{u}$ to a unit-range valuation profile $\mathbf{u''}$. By Lemma \ref{lem:RPopt}, the optimal allocation can be achieved by a run of random priority, so we know that in the optimal allocation at most $1$ agent will be matched with its least preferred item. Now consider the valuation profile $\mathbf{u}'$ where each agent $i$'s valuation for its least preferred item is set to 0 (unless it already is 0) and the rest of the valuations are as in $\mathbf{u}$. Since the ordinal preferences of agents are unchanged, random priority performs worse on this valuation profile, and because of Lemma \ref{lem:RPopt}, $w^{*}(\mathbf{u'}) \geq w^{*}(\mathbf{u})-1/n$. Next consider the valuation profile 
\[
\mathbf{u}'' = \begin{pmatrix}
    \mathbf{u}' & \mathbf{1} \cr
    \mathbf{o}^T & 1 \end{pmatrix}
 \] 
where $\mathbf{o} \in \RR^n$ and $\mathbf{o}_j = (j-1)/n^5$. That is, $\mathbf{u}''$ has $n+1$ agents and items, where agents $1,...,n$ have the same valuations for items $1,...,n$ as in $\mathbf{u}'$, every agent has a valuation of 1 for item $n+1$, and agent $n+1$ only has a significant valuation for item $n+1$. Notice that $\mathbf{u}''$ is a unit-range valuation profile, and $w^{*}(\mathbf{u''}) \geq  w^{*}(\mathbf{u'})+1$. Furthermore, $\EE\left[\sum_{i=1}^n u_i(RP(\mathbf{u'}))\right] \geq \EE\left[\sum_{i=1}^n u_i(RP(\mathbf{u''}))\right] -2$ and hence 
\begin{eqnarray*}
	\frac{\EE\left[\sum_{i=1}^n u_i(RP(\mathbf{u})_i)\right]}{w^{*}(\mathbf{u})} 
	&\geq& \frac{\EE\left[\sum_{i=1}^n u_i(RP(\mathbf{u'})_i)\right]}{w^{*}(\mathbf{u'}) + 1/n} 
	\geq \frac{\EE\left[\sum_{i=1}^n u_i(RP(\mathbf{u''})_i)\right] - 2}{w^{*}(\mathbf{u''}) + 1/n - 1} \\
	&\geq& \frac{\EE\left[\sum_{i=1}^n u_i(RP(\mathbf{u''})_i)\right]}{w^{*}(\mathbf{u''})} - \frac{2}{w^{*}(\mathbf{u''})}
	\geq \frac{C}{\sqrt{n}} -  \frac{2}{w^{*}(\mathbf{u})} \\
	&\geq& \frac{C}{\sqrt{n}} -  \frac{2}{4\sqrt{n} / C}
	= \frac{C}{2\sqrt{n}}.
\end{eqnarray*}
This completes the proof. \hfill $\square$
\end{proof}


The next lemma bounds the approximation ratio of any ordinal (not necessarily truthful-in-expectation) mechanism. 

\begin{lemma}\label{lem:ordinalupperUS}
Let $J$ be an ordinal mechanism for unit-sum representation. Then $ar(J) = O\left(n^{-1/2}\right)$.
\end{lemma}

\begin{proof}
Assume for ease of notation that $n$ is a square number; the proof can easily be adapted to the general case. By Lemma \ref{lem:anonymous}, we can assume without loss of generality that $J$ is anonymous. We will use the following valuation profile $\mathbf{u}$ where $\forall i \in \{1,...,\sqrt{n}\}$:
\begin{align*}
	u_i(j) &= \begin{cases} 
		1-\sum_{j\neq i} u_i(j), &\mbox{for } j = i , j \leq \sqrt{n} \\
		\frac{n-j}{10n^5}, &\mbox{otherwise }\end{cases}  \\
	u_{i+l  \sqrt{n}}(j) &= \begin{cases} 
		1-\sum_{j \neq i} u_i(j), &\mbox{for } j=i , j \leq \sqrt{n}\\
		\frac{1}{\sqrt{n}}-\frac{j}{10n^2}, &\mbox{for } j\neq i , j \leq \sqrt{n}  \\
		\frac{n-j}{10n^5}, &\mbox{otherwise } \end{cases} \ \ , \ \ l \in \{1,...,\sqrt{n}-1\} 
\end{align*}
Intuitively, $\mathbf{u}$ is a valuation profile where for each $1 \leq i \leq \sqrt{n}$, agent $i$'s valuation function induces the same ordering as agent $(i+l \cdot \sqrt{n})$'s valuation function, for $1 \leq l \leq \sqrt{n}-1$. For agent $i=1,...,\sqrt{n}$, because of anonymity, agent $i$ can at most expect to get a uniform lottery over all the items with each of the other $\sqrt{n}-1$ agents that have the same ordering of valuations. For agents $\sqrt{n}+1,\ldots,n$, the contribution to the social welfare from items $1,...,\sqrt{n}$ is at most $2$ since their valuations for these items are bounded by $2/\sqrt{n}$, and their contribution to the social welfare from items $\sqrt{n}+1,...,n$ is similarly bounded by 1. Thus we can write an upper bound on the expected welfare as:

\begin{equation*}
\sum_{i=1}^{\sqrt{n}} \EE\left[u_i(J(\mathbf{u})_i\right)] + \sum_{i=\sqrt{n}+1}^n \EE\left[u_i(J(\mathbf{u})_i)\right] 
\leq \sum_{i=1}^{\sqrt{n}} \frac{1}{\sqrt{n}} + 3 = 4, 
\end{equation*} 
while the social welfare of the optimal allocation is at least $\sqrt{n} - 1/10n^3$. From this, we get $ar(J) \leq 8/\sqrt{n}$. \hfill $\square$
\end{proof}
  
Finally, the upper bound for any truthful-in-expectation mechanism is given by the following lemma.

\begin{lemma}\label{lem:truthfulupperUS}
Let $J$ be a truthful-in-expectation mechanism for unit-sum representation. Then $ar(J) = O\left(n^{-1/2}\right)$.
\end{lemma}

\begin{proof}
Intuitively, the lemma is true because the valuation profile used in the proof of Lemma \ref{lem:truthfulupper} can be easily modified in a way such that all rows of the matrices of valuations sum up to one. Specifically, consider the following valuation profile:
\begin{align*}
u_i(j) &= \begin{cases} 1-\sum_{j\neq i} u_i(j), &\mbox{for } j = i \\
\frac{2}{10k}-\frac{j}{10n}, &\mbox{for } 1\leq j \leq k+1, j \neq i \\
\frac{n-j}{10n^2}, &\mbox{otherwise }\end{cases} \ \ &\forall i \in \{1,\ldots,k+1\} \\
u_i(j) &= \begin{cases} 1-\sum_{j \neq 1} u_i(j), &\mbox{for } j=1 \\
\frac{2}{10k}-\frac{j}{10n}, &\mbox{for } 1 < j \leq k+1 \\
\frac{n-j}{10n^2}, &\mbox{otherwise }
\end{cases}  \ \  &\forall i \in \{k+2,\ldots,n\}  
\end{align*}

Note that this is exactly the same valuation profile used in the proof of Lemma \ref{lem:truthfulupper} where all entries are divided by ten, except those where the valuation is $1$, which are now equal to $1$ minus the sum of the valuations for the rest of the items. This modification will only carry a factor of $1/10$ through the calculations and hence the proven bound will be the same asymptotically. \hfill $\square$
\end{proof}

\section{Extensions and special cases} \label{sec:future}

\subsection{Allowing ties}

Our results extend if we allow ties in the image of the valuation function. All of our upper bounds hold trivially. For the approximation guarantee of random priority, first the mechanism clearly must be equipped with some tie-breaking rule to settle cases where indifferences appear. For all natural (fixed before the execution of the mechanism) tie-breaking rules the lower bounds still hold. To see this, consider any valuation profile with ties and a tie-breaking rule for random priority. We can simply add sufficiently small quantities $\epsilon_{ij}$ to the valuation profile according to the tie-breaking rule and create a new profile without ties. The assignment probabilities of random priority will be exactly the same as for the version with ties, and random priority achieves an $\Omega(1/\sqrt{n})$ approximation ratio on the new profile. Then since $\epsilon_{ij}$ were sufficiently small, the same bound holds for the original valuation profile.

\subsection{$\mathbf{[0,1]}$ valuation functions}

All of our results apply to the extension of the unit-range representation where $0$ is not required to be in the image of the function, that is $\max_ju_i(j)=1$ and for all $j$, $u_i(j) \in [0,1]$. This representation captures scenarios where agents can be more or less indifferent between every single item.  Since every unit-range valuation profile is also a valid profile for this representation, the upper bounds hold trivially. For the approximation ratio of random priority, we obtain the following corollary. 
  
\begin{corollary}\label{col:upperunitsum}
The approximation ratio of random priority, for the setting with $[0,1]$ valuation functions is $\Omega\left(n^{-1/2}\right)$.
\end{corollary}

\begin{proof}
Let $\mathbf{u}$ be any $[0,1]$ valuation profile and let $C$ be the constant in the lower bound of Lemma \ref{lem:RPlower}. Similarly to the proof of Lemma \ref{lem:RPlowerUS}, notice that by Lemma \ref{lem:RPopt}, the optimal matching on $\mathbf{u}$ matches at most one agent with its least-preferred item. So let $\mathbf{u'}$ be the valuation profile in which each agent $i$ has the same valuation for every item as in profile $\mathbf{u}$, except the valuation for its least preferred item is set to $0$ (if it is not already $0$). Doing this, the expected social welfare of random priority becomes smaller, and $w^{*}(\mathbf{u'}) \geq w^{*}(\mathbf{u}) - 1$. Notice that $\mathbf{u'}$ is now unit-range and by Lemma \ref{lem:RPlower}, we get that

\begin{eqnarray*}
\frac{\EE\left[\sum_{i=1}^n u_i(RP(\mathbf{u})_i)\right]}{w^{*}(\mathbf{u})} 
	&\geq& \frac{\EE\left[\sum_{i=1}^n u_i(RP(\mathbf{u'})_i)\right]}{w^{*}(\mathbf{u'}) + 1}\\
	&\geq& \frac{\EE\left[\sum_{i=1}^n u_i(RP(\mathbf{u'})_i)\right]}{2w^{*}(\mathbf{u'})} \\
	&\geq& C/2\sqrt{n}
\end{eqnarray*}\hfill $\square$
\end{proof}

\subsection{Improved approximations for $n=3$}

Theorem \ref{thm:main} implies that random priority is indeed the best truthful-in-expectation mechanism for the problem, when only considering the asymptotic behavior of mechanisms. We now consider non-asymptotic behavior by studying the case when $n=3$ and present an non-ordinal mechanism that achieves better bounds than any ordinal mechanism, when the representation of the valuation functions is unit-range. 

Since $n=3$ and the representation is unit-range, the valuation function of an agent $i$ can be completely specified by a tuple $(ABC,\alpha_i)$ where $ABC$ is the ordering of items $1,2$ and $3$ and $\alpha_i$ is the valuation of agent $i$ for its second to most preferred item. For valuation profile, the valuation function $u_i(1)=0.6, u_i(2)=1, u_i(3)=0$ can be written as $(213,0.6)$. By this, we can generate all possible valuation profiles with three agents and three items using $\alpha_1,\alpha_2,\alpha_3$ as variables. By anonymity and neutrality, we can compress the search space drastically (by pruning symmetric profiles) and then calculate the ratios on all valuation profiles as functions of $\alpha_1,\alpha_2,\alpha_3$.

For the case of random priority, it is easy to see where those ratios are minimized and the approximation ratio is the worst ratio over all valuation profiles that we consider. It turns out that the approximation ratio of random priority for $n=3$ is $2/3$. In fact, random priority achieves the optimal approximation ratio among all ordinal mechanisms. To see this, observe that the worst-case ratio of random priority is given by the following ordered valuation profile:
\[
 \mathbf u = \left(\begin{array}{ccc}
1 & 1-\epsilon & 0 \\
1 & \epsilon & 0 \\
1 & \epsilon & 0
\end{array}\right) 
\]
Notice that when $\epsilon$ tends to $0$, the ratio of any ordinal mechanism on $\mathbf{u}$ tends to $2/3$. 

Next, consider the one-agent mechanism that given the reported valuation function matches the agent with its most preferred item with probability $(6-2\alpha^3)/8$, with its second to most preferred item with probability $(1+3\alpha^2)/8$ and with its least preferred item with probability $(1-3\alpha^2+2\alpha^3)/8$. This mechanism, that we will refer to as the \emph{cubic lottery} was presented in \cite{FT:10} and proven by the authors to incentivize the agent to report truthfully. Now consider the following mechanism for the one-sided matching problem:

\begin{mechanism}{(Hybrid mechanism - HM)}
Uniformly at random fix a permutation $\sigma \in S$ of the agents. Match agent $\sigma(1)$ with item $j \in \{1,2,3\}$ with probabilities given by the cubic lottery. Match agent $\sigma(2)$ with its favorite item from the set of still available items. Match agent $3$ with the remaining item.
\end{mechanism}

Since the permutation of agents is fixed uniformly at random, this mechanism is truthful-in-expectation. We prove the following theorem.

\begin{theorem}
$ar(HM) = 0.699$. 
\end{theorem} 

\begin{proof}
Observe that the mechanism is anonymous and neutral, hence we can follow the same procedure described above and generate all possible valuation profiles with $n=3$ and then prune the profile space to obtain a relatively small number of valuation profiles. The ratio on a valuation profile $\mathbf{u}$ will be a function of the form $G(\alpha_1,\alpha_2,\alpha_3)=g_1(\alpha_1,\alpha_2,\alpha_3)/g_2(\alpha_1,\alpha_2,\alpha_3)$ where $g_1:V_1 \times V_2 \times V_3 \rightarrow  \RR$ is a non-linear function corresponding to the expected social welfare and $g_2:V_1 \times V_2 \times V_3 \rightarrow \RR$ is a linear function corresponding to the maximum weight matching on $\mathbf{u}$. Then, to calculate the approximation ratio, we need to solve a non-linear program of the form ``minimize $G(\alpha_1,\alpha_2,\alpha_3)$ subject to $\alpha_1,\alpha_2,\alpha_3 \in [0,1]$'' for every valuation profile. The minimum over all valuation profiles is the approximation ratio of the mechanism. We use standard non-linear programming software to obtain the bound; we consider such a computer-assisted proof sufficient for the purpose of this subsection.  \hfill $\square$
\end{proof}

Notice that approximation ratio achieved by the hybrid mechanism is strictly larger than the approximation ratio of any ordinal mechanism.  The next question would be whether we can prove similar bounds for other (small) values of $n$. We might be able to extend the technique used above to $n=4$ by relying heavily on computer-assisted programs to generate the valuation profiles and calculate the ratios but it would be difficult to extend it to any larger number of agents, since the valuation profile space becomes quite large. A different approach for proving approximation guarantees for concrete values of $n$ would be interesting. Finally, it would be interesting to investigate whether random priority obtains the (non-asymptotic) optimal approximation ratio among ordinal mechanisms for all values of $n$.

\newpage

\bibliographystyle{plain}
\bibliography{RP}

\begin{thebibliography}{10}

\bibitem{AS:13}
Atila Abdulkadiroglu and Tayfun S{\"o}nmez.
\newblock Matching markets: Theory and practice.
\newblock {\em Advances in Economics and Econometrics (Tenth World Congress)},
  pages 3--47, 2013.

\bibitem{AN:10}
Elliot Anshelevich and Sanmay Das.
\newblock Matching, cardinal utility, and social welfare.
\newblock {\em ACM SIGECom Exchanges}, 9(1):4, 2010.

\bibitem{BA:10}
Salvador Barbera.
\newblock Strategy-proof social choice.
\newblock In K.~J. Arrow, A.~K. Sen, and K.~Suzumura, editors, {\em Handbook of
  Social Choice and Welfare}, volume~2, chapter~25. North-Holland: Amsterdam,
  2010.

\bibitem{BCK:11}
Anand Bhalgat, Deeparnab Chakrabarty, and Sanjeev Khanna.
\newblock Social welfare in one-sided matching markets without money.
\newblock In {\em APPROX-RANDOM}, pages 87--98, 2011.

\bibitem{BM:01}
Anna Bogomolnaia and Herv\'{e} Moulin.
\newblock A new solution to the random assignment problem.
\newblock {\em Journal of Economic Theory}, 100:295--328, 2001.

\bibitem{BCHLPS:12}
Craig Boutilier, Ioannis Caragiannis, Simi Haber, Tyler Lu, Ariel~D. Procaccia,
  and Or~Sheffet.
\newblock Optimal social choice functions: a utilitarian view.
\newblock In {\em Proceedings of the 13th ACM Conference on Electronic
  Commerce}, pages 197--214. ACM, 2012.

\bibitem{DG:10}
Shaddin Dughmi and Arpita Ghosh.
\newblock Truthful assignment without money.
\newblock In {\em ACM Conference on Electronic Commerce}, pages 325--334, 2010.

\bibitem{FT:10}
Uriel Feige and Moshe Tennenholtz.
\newblock Responsive lotteries.
\newblock In {\em SAGT}, pages 150--161, 2010.

\bibitem{FM:13}
Aris Filos-Ratsikas and Peter~Bro Miltersen.
\newblock Truthful approximations to range voting.
\newblock {\em CoRR}, abs/1307.1766, 2013.

\bibitem{GC:10}
Mingyu Guo and Vincent Conitzer.
\newblock Strategy-proof allocation of multiple items between two agents
  without payments or priors.
\newblock In {\em Proceedings of the 9th International Conference on Autonomous
  Agents and Multiagent Systems: volume 1-Volume 1}, pages 881--888, 2010.

\bibitem{HZ:79}
Aanund Hylland and Richard Zeckhauser.
\newblock The efficient allocation of individuals to positions.
\newblock {\em The Journal of Political Economy}, 87(2):293--314, 1979.

\bibitem{PP:10}
Ariel~D. Procaccia.
\newblock Can approximation circumvent gibbard-satterthwaite?
\newblock In {\em AAAI}, 2010.

\bibitem{PT:09}
Ariel~D. Procaccia and Moshe Tennenholtz.
\newblock Approximate mechanism design without money.
\newblock In {\em ACM Conference on Electronic Commerce}, pages 177--186, 2009.

\bibitem{SU:11}
Tayfun S{\"o}nmez and Utku {\"U}nver.
\newblock Matching, allocation and exchange of discrete resources.
\newblock {\em Handbook of Social Economics}, 1A:781--852, 2011.

\bibitem{SVE:99}
Lars-Gunnar Svensson.
\newblock Strategy-proof allocation of indivisble goods.
\newblock {\em Social Choice and Welfare}, 16(4):557--567, 1999.

\bibitem{ZHOU:90}
Lin Zhou.
\newblock On a conjecture by gale about one-sided matching problems.
\newblock {\em Journal of Economic Theory}, 52:123--135, 1990.

\end{thebibliography}


\end{document}